\documentclass{article}
\usepackage{graphicx}
\usepackage{amsmath,amsthm,amssymb,amsfonts}
\usepackage{verbatim}
\usepackage{stmaryrd} %for \leftrightarroweq
\usepackage{hyperref}

\usepackage[colorinlistoftodos]{todonotes}

\theoremstyle{plain}
\newtheorem{thm}{Theorem}

\newtheorem{cor}[thm]{Corollary}

\theoremstyle{definition}

\numberwithin{thm}{section}

\begin{document}
\title{Coarse Grained Parallel Selection}
\author{Laurence Boxer
         \thanks{
    Department of Computer and Information Sciences,
    Niagara University,
    Niagara University, NY 14109, USA;
    and Department of Computer Science and Engineering,
    State University of New York at Buffalo.
    E-mail: boxer@niagara.edu
    }
}

\date{ }
\maketitle

\begin{abstract}
Several efficient, but non-optimal, solutions to the Selection Problem on coarse grained
parallel computers have appeared in the literature. We consider the example of 
the Saukas-Song algorithm; we analyze it without expressing the running time in terms of 
communication rounds. This shows that while in the best case the Saukas-Song algorithm 
runs in asymptotically optimal time, in general it does not. We propose another 
algorithm for coarse grained selection that has optimal expected running time.


\end{abstract}
%\twocolumn

\section{Introduction}
Several algorithms for solution of the Selection Problem on coarse grained
parallel computers, including those of~\cite{AlFur,Bader,Fujiwara,Gerbessiotis,Ishimizu,SS99,Tiskin},
have been proposed. Among these are algorithms that are efficient but not
asymptotically optimal.

We consider for example the paper~\cite{SS99}, by Saukas and Song, who
present the analysis of their algorithm in terms of the amount of time spent in
local sequential operations and the number of communications rounds. 
In the current paper, we replace analysis of the number of communications rounds
with an analysis of their running times, giving us asymptotic analysis of the 
running time for the entire algorithm. This lets us show that although
the Saukas-Song algorithm is efficient, it is not optimal. 
We show that an algorithm for coarse grained parallel
selection that is not asymptotically optimal can be modified to obtain an algorithm
with asymptotically optimal expected running time.

\section{Preliminaries}
\subsection{Model of Computation}
Material in this section is quoted or paraphrased from~\cite{BMR99}.

The {\em coarse grained multicomputer} model, or $CGM(n,p)$, considered
in this paper, has $p$ processors with $\Omega(n/p)$ local memory
apiece - i.e., each processor has $\Omega(n/p)$ memory cells of 
$\Theta(\log n)$ bits apiece. The processors may be connected to
some (arbitrary) interconnection network (such as a mesh, hypercube,
or fat tree) or may share global memory. A processor may exchange
messages of $O(\log n)$ bits with any immediate neighbor
in constant time. In determining time complexities, we consider
both local computation time and interprocessor communication
time, in standard fashion. The term ``coarse grained'' means that
the size $\Omega(n/p)$ of each processor's memory is
``considerably larger'' than $\Theta(1)$; by convention, we usually assume
$n/p \ge p$ (equivalently, $n \ge p^2$), but will occasionally assume other relations
between $n$ and $p$, typically such that each processor has
at least enough local memory to store the ID number of every other
processor. For more information on this model and associated
operations, see~\cite{DFR96}.

\subsection{Terminology and notation}
We say an array $list[1 \ldots n]$ is {\em evenly distributed} in a $CGM(n,p)$
if each processor has $\Theta(n/p)$ members of the array.

We use the notation $\#X$ for the number of members of the set~$X$.

\subsection{Semigroup operations}

Let $X = \{x_j\}_{j=1}^n$ be a set of data values and let
$\circ$ be a binary operation on $X$. A semigroup operation computes
$x_1 \circ x_2 \circ \ldots \circ x_n$. Examples of such operations include
{\em sum, average, min}, and {\em max}.
%A {\em parallel prefix} operation computes all of
%\[ x_1, ~ x_1 \circ x_2, ~ \ldots, ~ x_1 \circ x_2 \circ \ldots \circ x_n.
%\]

\begin{thm}
{\rm \cite{BM04}}
\label{semigroup}
Let $X = \{x_j\}_{j=1}^n$ be a set of data values distributed $\Theta(n/p)$
per processor in a $CGM(n,p)$. Then the semigroup computation of
$x_1 \circ x_2 \circ \ldots \circ x_n$
can be performed in $\Theta(n/p)$ time. At the end of this algorithm,
all processors hold the value of $x_1 \circ \ldots \circ x_n$. $\Box$
\end{thm}
\begin{comment}
\begin{thm}
{\rm \cite{BM04}}
\label{parallelPref}
Let $X = \{x_j\}_{j=1}^n$ be a set of data values distributed $\Theta(n/p)$
per processor in a $CGM(n,p)$. Then the parallel prefix computation of
\[ x_1, ~ x_1 \circ x_2, ~ \ldots, ~ x_1 \circ x_2 \circ \ldots \circ x_n.
\]
can be performed in $\Theta(n/p)$ time. At the end of this algorithm,
the processor that holds $x_j$ also holds $x_1 \circ \ldots \circ x_j$. $\Box$
\end{thm}
\end{comment}

\subsection{Data movement operations}
In the literature of CGM algorithms, many papers, including~\cite{SS99},
analyze an algorithm by combining the running time of sequential computations
with the number of {\em communications rounds}.  A communications round is described
in~\cite{SS99} as an operation in which each processor of a
$CGM(n,p)$ can exchange $O(n/p)$ data with other processors.

However, this definition does not lead to a clear understanding of the 
asymptotic running time of a CGM algorithm. E.g., a communication round could
require a processor to send $\Theta(1)$ data to a neighboring processor, which can be 
done in $\Theta(1)$ time; or, a communication round could require communication of $\Theta(1)$ 
data between diametrically opposite processors of a linear array, which
requires $\Theta(p)$ time.

Further, there is a sense in which the notion of a communication round is not well 
defined. Consider again the example of communication of $\Theta(1)$ data
between diametrically opposite processors $P_1, P_p$ of a linear array in which 
the processors $P_i$ are indexed sequentially, i.e., $P_1$ is adjacent to $P_2$ and $P_i$ 
is adjacent to $P_{i-1}$ and to $P_{i+1}$ for $1 < i < p$.
This communication can be regarded as a single communication round according to the
description given above; or as $p-1$ communication rounds, in the $i^{th}$ of which processor
$P_i$ sends $\Theta(1)$ data to processor $P_{i+1}$, $1 \le i < p$.

Now that more is known about the running times of communications operations in CGM than
when~\cite{SS99} appeared, we can fully analyze the running time of the algorithm. In the remainder of this
section, we discuss running times of communications operations used in the Saukas-Song algorithm.

\begin{thm}
\label{broadcast}
{\rm \cite{BM04}}
A unit of data can be broadcast from one processor to all other processors of a $CGM(n,p)$ in
$O(p)$ time. $\Box$
\end{thm}

Let $S$ be a set of data values distributed, not necessarily evenly, among the processors of
a parallel computer. A {\em gather} operation results in a copy of
$S$ being in a single processor $P_i$, and we say $S$ has been
gathered to $P_i$~\cite{BM04,MiSt96}. We have the following.

\begin{thm}
\label{gather}
{\rm \cite{BM04}}
Let $S$ be a nonempty set of $N$ elementary data values distributed
among the processors of $G$, a $CGM(n,p)$, such that $N=\Omega(p)$ and
$N=O(n/p)$. Then $S$ can be gathered to any processor of $G$ in
$\Theta(N)$ time. $\Box$
\end{thm}

\begin{comment}
Let $S$ be a set of data values distributed one per processor in
a parallel computer. A {\em multinode broadcast} operation results in
every processor having a copy of $S$.

\begin{thm}
{\rm \cite{BM04}}
\label{multinode}
Let $S$ be a set of elementary data distributed one per processor in a
$CGM(n,p)$. Then a multinode broadcast operation, at the end of which
every processor has a copy of $S$, can be performed in $O(p^2)$ time.
\end{thm}
\end{comment}

\section{Analysis of the Saukas-Song algorithm}
The algorithm of~\cite{SS99} is given in Figure~\ref{alg-fig}. Note it is assumed
in~\cite{SS99} that $n > p^2 \log p$. We will refer
to the steps of the algorithm as labeled in this figure. For convenience, we will take
$c=1$ in step~(2).

%\onecolumn
\begin{figure}
\includegraphics[height=3.25in]{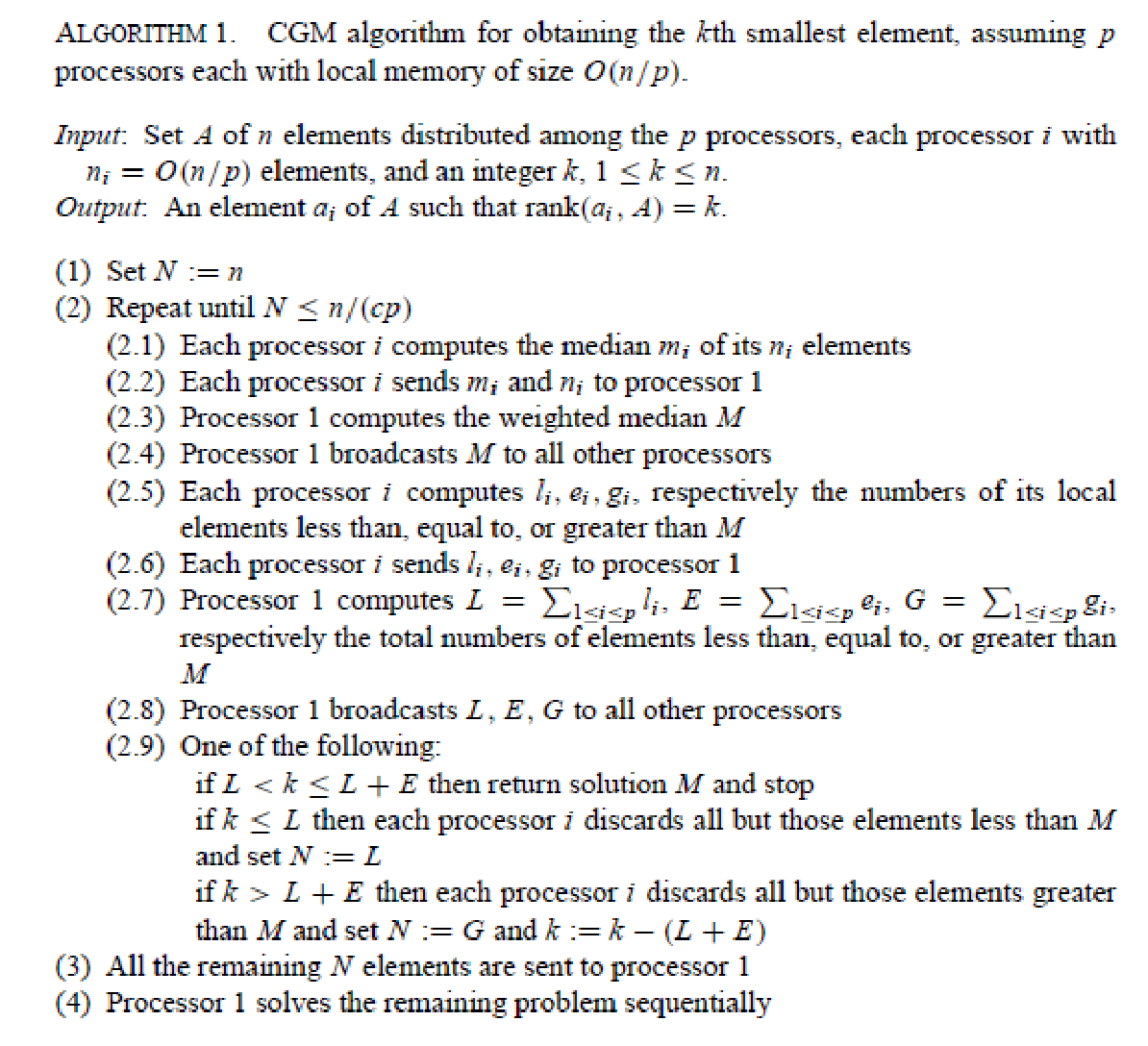}
\caption{The Saukas-Song algorithm of~\cite{SS99}}
\label{alg-fig}
\end{figure}

%\twocolumn

\begin{thm}
{\rm \cite{SS99}}
\label{reduction}
After each performance of the body of the loop in step~(2), the number of
elements remaining under consideration is reduced by at least one quarter. $\Box$
\end{thm}

\begin{cor}
\label{performances}
The number of performances of the body of the loop in step~(2) is $O(\log p)$.
\end{cor}

\begin{proof}
It follows from Theorem~\ref{reduction} that after $k$ performances of the body of
the loop, the number of elements remaining under consideration is at most
$(3/4)^k n$. Since step~(2) terminates in the worst case when the number
of elements remaining under consideration is at most $n/p$, termination requires,
in the worst case, $(3/4)^k n \le n/p$, or $p \le (4/3)^k$. The smallest integer $k$ 
satisfying this inequality must therefore satisfy $k = \Theta(\log p)$. Since the 
loop could terminate after as little as 1 performance of its body, the assertion follows.
\end{proof}

\begin{thm}
\label{S-S}
Assume $n > p^2 \log_2 p$.
The Saukas-Song algorithm for the Selection Problem runs in
$\Theta(\frac{n \log p}{p})$ time on a $CGM(n,p)$ in the worst case. In the best case, the running time
is $\Theta(n/p)$. We may assume that at the end of the algorithm, every processor has the solution to
the Selection Problem.
\end{thm}

\begin{proof}
We give the following argument.

\begin{itemize}
\item Clearly, step~(1) executes in $\Theta(1)$ time.

\item We analyze step~(2) as follows.
\begin{enumerate}
\item Step~(2.1) is executed by a linear-time sequential
      algorithm~\cite{Blum-etal,MB}. As Saukas and Song
      observed, our algorithm does not guarantee that the data being considered are
      evenly distributed among the processors throughout the repetitions of the loop
      body. In the worst case, some processor could have $\Theta(n/p)$ data in each
      iteration of the loop body. Therefore, this step executes in worst case
      $\Theta(n/p)$ time.
\item Step~(2.2) is performed by a gather operation. By Theorem~\ref{gather},
      this can be done in $\Theta(p)$ time.
\item Step~(2.3) is performed in $\Theta(p)$ time by a linear-time sequential
      algorithm.
\item Step~(2.4) is performed by a broadcast operation. By Theorem~\ref{broadcast}, this requires $O(p)$ time.
\item Step~(2.5) is performed by sequential semigroup (counting) operations performed
      by all processors in parallel, in linear time~\cite{MB}.
      As noted above, in the worst case a processor 
      could have $\Theta(n/p)$ data in each iteration of the loop body. 
      Therefore, this step executes in $O(n/p)$ time.
\item Step~(2.6) is performed by a gather operation. As above, this can be done in 
      $O(p)$ time.
\item Step~(2.7) is executed by semigroup operations in $\Theta(p)$ time.
\item Step~(2.8) is performed by a broadcast operation in $O(p)$ time.
\item In the worst case, Step~(2.9) requires each processor $P_i$ to discard $O(n/p)$ data items. This can be
      done as follows. In parallel, each processor $P_i$ rearranges its share of the data so that the undiscarded
      items are at the beginning of the segment of (a copy of) $list$ stored in $P_i$, using a sequential
      prefix operation in $O(n/p)$ time.
\end{enumerate}
Thus, the worst case time required for one performance of the body of the loop of step~(2) is
$\Theta(n/p + p) = \Theta(n/p)$. By Corollary~\ref{performances},
the loop executes its body $\Theta(\log p)$ times in the worst case. Thus, the loop
executes all performances of its body in worst case $\Theta(\frac{n \log p}{p})$ time.
\item Step~(3) is performed by a gather operation. Since we now have
      $N \le n/p$, Theorem~\ref{gather} implies this is done in $O(n/p)$ time.
\item Step~(4) uses a linear time sequential algorithm to solve the problem in
      $\Theta(N)=O(n/p)$ time.
\item Additionally, processor 1 can broadcast its solution to all other processors. By Theorem~\ref{broadcast},
      this requires $O(p)$ time.
\end{itemize}

Thus, the algorithm uses worst case $\Theta(\frac{n \log p}{p})$ time.

In the best case, the loop of step~(2) executes its body once, when in step~(2.9)
it is found that $L < k \le L+E$. In this case, step~(2) executes in
$\Theta(n/p)$ time, and the algorithm executes in $\Theta(n/p)$ time.
\end{proof}

\section{Expected-optimal selection algorithms}
In this section, we consider modification of %, although somewhat related, 
algorithms for the selection problem on coarse grained multicomputers,
based on the familiar idea that if we know the distribution of data values, then,
with high probability, we can estimate the desired value to within a small interval and,
in so doing, eliminate (with high probability) most of the data values from consideration
as possible solutions. In the following, we assume the data values are uniformly
distributed over some interval.

\subsection{Useful formulas from probability and statistics}
We use the notation $Pr[H]$ for the probability of the event $H$, and
$Pr[D|C]$ for the conditional probability of the event $D$ given that $C$ occurs. We use
the notation $\overline{H}$ for ``not~$H$''.

Let $X_{(k)}$ be a random variable for the $k^{th}$ smallest member of a set $\{x_i\}_{i=1}^n$ of data
points uniformly distributed over the interval $[0,1]$. Then the expected value $E(X_{(k)})$
and variance $V(X_{(k)})$ of $X_{(k)}$
are given~\cite{Rice95} by the following.
\begin{equation}
\label{mean}
E(X_{(k)}) = k / (n+1)
\end{equation}
\begin{equation}
\label{variance}
V(X_{(k)}) = \frac{k(n+1-k)}{(n+1)^2 (n+2)}
\end{equation}

We recall also Chebyshev's Inequality:

\begin{thm}
{\rm ~\cite{Rice95}}
\label{Cheb}
Let $y$ be a random variable with expected value $E(y)=\mu$ and variance $V(y)=v$. Then, for any $t>0$,
\begin{equation}
\label{Cheb-ineq}
Pr \left[|y - \mu| > t \right] \le \frac{v}{t^2}.
\end{equation}
\end{thm}

\subsection{Algorithm}
The solution to the Selection Problem proposed below assumes an existing solution $F$,
such as the Saukas-Song algorithm,
and therefore must satisfy the same or a stronger restriction on the number of processors
as the solution $F$.
We assume there is a function $B_F(p)$ associated with $F$ such that
$p^2 \le B_F(p) \le n$.
E.g., for the Saukas-Song algorithm we take $B_F(p) = p^2 \log_2 p$.
We will use a slightly stricter limitation,
\begin{equation}
\label{processorRestriction}
\mbox{For some } w > 2, ~B_F(p) = p^w \le n; \mbox{ equivalently, } n/p \ge n^{1 - 1/w}.
\end{equation}
The latter inequality derives from the former, since
\[ p^w \le n \Leftrightarrow p \le n^{1/w} \Leftrightarrow n/p \ge \frac{n}{n^{1/w}} = n^{1 - 1/w}.
\]
Note this also implies
\begin{equation}
    \label{pAndn/p}
    p \le n^{1/w} \le \mbox{ ~~~(since } w >2)~~~~n^{1-1/w} \le n/p.
\end{equation}

We use the fact that the real function $G(x) = x(n-x)$ is maximized at $x=n/2$, from which it
follows that
\begin{equation} 
\label{quadraticBound}
   x(n-x) \le G(n/2) =  n^2 / 4.
\end{equation}
Our algorithm has an
asymptotically optimal expected running time, $\Theta(n/p)$, 
and its worst case running time is that of $F$. 

\begin{thm}
Let $F$ be an algorithm to solve the Selection Problem on a $CGM(n,p)$, where
$p$ is restricted by~(\ref{processorRestriction}). Suppose $F$ has running time
$R_F(n,p) = O(\frac{n \log p}{p})$ (e.g., the Saukus-Song algorithm).
Suppose we are given a set $A$ of $n$ elements distributed $\Theta(n/p)$ per processor 
among the processors of a $CGM(n,p)$ 
and an integer $k$ such that $1 \le k \le n$.
Assume the key values of the elements of $A$
are uniformly distributed over an interval $[u,v]$. Then the $k^{th}$ smallest member of $A$ can be
found in expected $\Theta(n/p)$ time. The worst case running time is $\Theta(R_F(n,p))$.
\end{thm}

\begin{proof}
Without loss of generality, $[u,v]=[0,1]$.

If the values of $n$ and $p$ are not known, they can be computed and made know to all processors
via semigroup (counting) operations. By Theorem~\ref{semigroup}, this can be done in $\Theta(n/p)$ time.

Let $c$ be a small positive integer.

If $k \le c$ then we can find the desired result by $k$ performances of a minimum computation. I.e., we
do the following.
\begin{tabbing}
For \= $i=1$ to $k$  \\
\> Find an element $a_j \in A$ such that $a_j = \min\{a \in A\}$.\\
\> If $i=k$ \= then the desired result is $a_j$;
   else set $A = A \setminus\{a_j\}$.\\
End For
\end{tabbing}

By Theorem~\ref{semigroup} and~(\ref{pAndn/p}), the For loop executes 
in $\Theta(n/p)$ time, and by Theorem~\ref{broadcast},
the result can be distributed to all processors in $O(p)$ time. By~(\ref{pAndn/p}),
the case $k \le c$ is solved in $\Theta(n/p)$ time.

If $k \ge n-c$, then we do the following.
\begin{tabbing}
For \= $i=n$ to $k$ step $-1$ \\
\> Find an element $a_j \in A$ such that $a_j = \max\{a \in A\}$. \\
\> If $i=k$ then the desired result is $a_j$; else set 
   $A = A \setminus\{a_j\}$.\\
End For
\end{tabbing}

By Theorem~\ref{semigroup} and~(\ref{pAndn/p}), the For loop executes in $\Theta(n/p)$ time, and by 
Theorem~\ref{broadcast}, the result can be distributed to all processors in $O(p)$ time.
By~(\ref{pAndn/p}), the case $k \ge n-c$ is solved in $\Theta(n/p)$ time.

The remaining case is $c < k < n-c$.
\begin{comment}
Let $d$ be a constant such that $0.5 < d < 1$. Note 
\[p < (B_F(p))^{1/2} < n^{1/2} < n^d, \]
hence
\begin{equation}
n/p > n/n^d = n^{1-d}.
\label{pANDd}
\end{equation}
Let $\varepsilon = 1 - d > 0$.
\end{comment}
Let
\begin{equation}
    \label{epsiEq}
    \varepsilon = n^{-1/3}.
\end{equation}
Proceed as follows.
\begin{itemize}
\item In $\Theta(1)$ time, each processor computes
\begin{equation}
\label{leftBound}
U = \max\{0, \frac{k}{n+1} - \varepsilon\}
\end{equation}
and
\begin{equation}
\label{rightBound}
V = \min\{1, \frac{k}{n+1} + \varepsilon\}.
\end{equation}
\item In parallel, each processor $P_j$ scans its portion of the data set $A=\{x_i\}_{i=1}^n$ to
      determine 
      \begin{itemize}
      \item $S_j$, the number of elements of $A$ in $P_j$ that are less than U; and
      \item $M_j$, the number of elements of $A$ in $P_j$ that are in $[U,V]$.
      \end{itemize}
      This can be done in $\Theta(n/p)$ time.
\item Gather the $S_j$ values to one processor and compute their sum, $S=\sum_{j=1}^p S_j$. Similarly,
      gather the $M_j$ values to one processor and compute their sum, $M=\sum_{j=1}^p M_j$. Broadcast
      the values $S,M$ to all processors. From Theorems~\ref{gather} and~\ref{broadcast}, we conclude
      that all this can be done in $\Theta(p)$ time.
\item If $S < k \le S + M$, then we have $X_{(k)} \in [U,V]$; and if further we have $M \le n/p$
      then the $k^{th}$ smallest member of $A$ is the $(k-S)^{th}$ smallest member of the subset
      $M'$ of $A$ consisting of members of $A \cap [U,V]$. Since $\#M'=M$, we can gather the elements
      of $M'$ to one processor in $\Theta(M)=O(n/p)$ time according to Theorem~\ref{gather}, and have that
      processor sequentially find the $(k-S)^{th}$ smallest member of $M'$ in $\Theta(M)=O(n/p)$ time,
      and broadcast the result to all processors in $O(p)$ time.
      
      Thus, if this case is realized, the running time of the algorithm is $\Theta(n/p)$,
      which is asymptotically optimal, since the optimal
sequential running time for solving the Selection Problem is $\Theta(n)$.
      
      This leaves the following cases to be considered.
      \begin{itemize}
      \item If $S < k \le S + M$ and $M > n/p$, then apply the algorithm~$F$ to 
            finding the $(k-S)^{th}$ smallest member of
             $T_1 = \{x \in A \, | \, S < x < S+M \}$. In the worst case, 
             $\#T_1 = \Theta(n)$ and the asymptotic running
            time of this step is $R_F(n,p)$.
      \item If $k \le S$ then apply the algorithm~$F$ to 
            finding the $k^{th}$ smallest member of
             $T_2 = \{x \in A \, | \, x < U \}$. In the worst case, 
             $\#T_2 = \Theta(n)$ and the asymptotic running
            time of this step is $R_F(n,p)$.
      \item Otherwise, $k > S+M$ and the desired value is the $(k-(S+M))^{th}$ smallest
            member of $T_3= \{ x \in A \, | \, x > V \}$.
            Use the algorithm~$F$ to find
            the $(k-(S+M))^{th}$ smallest member of $T_3$. In the worst case, $\#T_3 = \Theta(n)$ 
            and the asymptotic running time is $R_F(n,p)$.
      \end{itemize}
\end{itemize}

It remains for us to show that $\Theta(n/p)$ is the expected running time of the algorithm.

Let $C$ be the event
\[ C = [U \le X_{(k)} \le V].\] 
The event complementary to $C$ is
\[ \overline{C} = \left [ X_{(k)} < U  \right] \cup 
    \left[V < X_{(k)} \right].
\]
By~(\ref{leftBound}) and~(\ref{rightBound}), 
$\overline{C} \subset \left[|X_{(k)} - E(X_{(k)})| > \varepsilon \right]$,
so by~(\ref{mean}), (\ref{variance}), Theorem~\ref{Cheb}, (\ref{quadraticBound}), and (\ref{epsiEq}),
\begin{equation}
\label{Cbar}
Pr[\overline{C}] \le \frac{\frac{k(n+1-k)}{(n+1)^2 (n+2)}}{\varepsilon^2} \le
       \frac{(n+1)^2}{4 \varepsilon^2 (n+1)^2 (n+2)} = O(n^{-1/3}).
\end{equation}
Therefore,
\begin{equation}
\label{C-prob}
\lim_{n \to \infty} Pr[C] = 1.
\end{equation}

Let 
\[ D = \left [ \# \{ x \in A \, | \, U \le x \le v \} \le \frac{n}{2p} \right ].
\]
Notice
\begin{equation}
\label{faster}
    C \cap D \mbox{ is a subset of the set of occurrences for which we do not run } F.
\end{equation}
We have
\[ \overline{D} = \left[ \# \{x \in A \, | \, U \le x \le V \} > \frac{n}{2p}\right].\]
We consider the following cases.

\begin{itemize}
\item Suppose $0 < U$ and $V < 1$. Then
\[ [\overline{D} | C] \subset [X_{(k-\frac{n}{4p})} > U] \, \cup \, [X_{(k+\frac{n}{4p})} < V] =
\]
\[ \left [X_{(k-\frac{n}{4p})} - E(X_{(k-\frac{n}{4p})}) > U - E(X_{(k-\frac{n}{4p})}) \right ] 
    \, \cup \, 
    \left [ E(X_{(k+\frac{n}{4p})}) - V < E(X_{(k+\frac{n}{4p})}) - X_{(k+\frac{n}{4p})} \right ]
\]
\[  \subset \left [ |X_{(k-\frac{n}{4p})} - E(X_{(k-\frac{n}{4p})})| > U - \frac{k-\frac{n}{4p}}{n+1} \right ] 
    \, \cup \, 
    \left [ \frac{k+\frac{n}{4p}}{n+1} - V < | X_{(k+\frac{n}{4p})} - E(X_{(k+\frac{n}{4p})})| \right ].
\]
For large $n$, it follows from~(\ref{processorRestriction}), 
(\ref{epsiEq}), (\ref{leftBound}) and~(\ref{rightBound}) that 
\[
U - \frac{k-\frac{n}{4p}}{n+1} = \frac{n}{4p(n+1)} - \ n^{-1/3} = \frac{k+\frac{n}{4p}}{n+1} - V >0, 
\]
Therefore, Theorem~\ref{Cheb} yields
\[ Pr[\overline{D} | C] \le \frac{(k-\frac{n}{4p})(n+1-(k-\frac{n}{4p}))}{(n+1)^2 (n+2) (U - \frac{k-\frac{n}{4p}}{n+1})^2} +
\frac{(k+\frac{n}{4p})(n+1-(k+\frac{n}{4p}))}{(n+1)^2 (n+2) (\frac{k+\frac{n}{4p}}{n+1} - V)^2} \le
\]
(by (\ref{quadraticBound}))
\[  \frac{(n+1)^2}{4 (n+1)^2 (n+2) \left (\frac{n}{4p(n+1)} - \ n^{-1/3}  \right )^2} +
    \frac{(n+1)^2}{4 (n+1)^2 (n+2) \left (\frac{n}{4p(n+1)} - n^{-1/3}  \right )^2}
\]
\[ = \frac{2}{4(n+2)\left (\frac{n}{4p(n+1)} - n^{-1/3}  \right )^2} < ~~~
    \mbox{ (by (\ref{processorRestriction})) } ~~~
     \frac{1}{2(n+2)\left (\frac{n^{1-1/w}}{4(n+1)} - n^{-1/3}  \right )^2} 
\]
\[= O\left (n^{-(1-2/w)} \right ).
\]

\item If $U=0$, then by~(\ref{leftBound}), $\frac{k}{n+1} \le \varepsilon$, so 
by~(\ref{epsiEq}) and~(\ref{rightBound}), $V \le 2 \varepsilon$.
\[ \overline{D} = \left [ \# \{x \in A \, | \, x \le V \} > \frac{n}{2p} \right ] =
   \left [ X_{(\frac{n}{2p} + 1)} \le V \right ] \subset
   \left [ X_{(\frac{n}{2p} + 1)} \le 2 \varepsilon \right ] =
\]
\[ \left [ E(X_{(\frac{n}{2p} + 1)}) - 2 \varepsilon \le  E(X_{(\frac{n}{2p} + 1)}) - X_{(\frac{n}{2p} + 1)} \right ] \subset
\]
\[ \left [ E(X_{(\frac{n}{2p} + 1)}) - 2 \varepsilon \le | E(X_{(\frac{n}{2p} + 1)}) - X_{(\frac{n}{2p} + 1)} | \right ].
\]
Since we have
\[ E(X_{(\frac{n}{2p} + 1)}) - 2 \varepsilon = \frac{(\frac{n}{2p} + 1)}{n+1} - 2n^{-1/3} \ge 
\mbox{~~~(by (\ref{processorRestriction}))}
\]
\[\frac{(\frac{n^{1/2}}{2} + 1)}{n+1} - 2n^{-1/3} = \Theta(n^{-1/2}),
\] 
therefore, for large $n$, $E(X_{(\frac{n}{2p} + 1)}) - 2 \varepsilon > 0$, so by
Theorem~\ref{Cheb} we have
\[ Pr[\overline{D}] \le Pr \left [ E(X_{(\frac{n}{2p} + 1)}) - 2 \varepsilon \le | E(X_{(\frac{n}{2p} + 1)}) - X_{(\frac{n}{2p} + 1)} | \right ] \le
\]
\[ \frac{(\frac{\frac{n}{2p}+1}{n+1}) (n+1 - (\frac{\frac{n}{2p}+1}{n+1}))}{(n+1)^2(n+2) \left (\frac{(\frac{n}{2p} + 1)}{n+1} - 2n^{-1/3} \right )^2}.
\]
By (\ref{quadraticBound}), the numerator is at most $(n+1)^2/4$, so
\[ Pr[\overline{D}] \le \frac{1}{4(n+2)  \left (\frac{(\frac{n}{2p} + 1)}{n+1} - 2n^{-1/3} \right )^2} 
\]
\[ \mbox{ (by (\ref{processorRestriction})) } ~~~
   \frac{1}{4(n+2)  \left (\frac{(\frac{n^{1-1/w}}{2} + 1)}{n+1} - 2n^{-1/3} \right )^2} =
   O \left (  n^{-(1-2/w)} \right ).
\]
Hence 
\[ Pr[\overline{D}|C] = Pr[C \cap \overline{D}]/Pr[C] \le Pr[\overline{D}] / Pr[C] =
  O \left (  \frac{n^{-(1-2/w)}}{1} \right ) = 
\]
\[ O \left (  n^{-(1-2/w)} \right ).
\]
\item The case $V=1$ is symmetric with the case $U=0$ and similarly produces the result
      $Pr[\overline{D}|C]  =  O \left (  n^{-(1-2/w)} \right )$.
\end{itemize}
Thus, in all cases, 
\begin{equation} 
\label{DbarGivenC}
Pr[\overline{D}|C] =  O \left (  n^{-(1-2/w)} \right ).
\end{equation}
Therefore,
\begin{equation}
\label{allCasesDgivenC}
Pr[D|C] = 1 - Pr[\overline{D}|C] \to_{n \to \infty} 1.
\end{equation}
From~(\ref{C-prob}) and~(\ref{allCasesDgivenC}),
\[ Pr[C \cap D] = Pr[C] \, Pr[D|C] \to_{n \to \infty} 1.
\]
Let $R(A)$ denote the running time of this algorithm for the event A.
By~(\ref{faster}), $R(C \cap D) = \Theta(n/p) = O(R_F(n,p))$, and
the expected running time $T(n,p)$ of the algorithm is
\[ E(T(n,p)) = R(C \cap D) Pr[C \cap D] + R(\overline{C \cap D}) Pr[\overline{C \cap D}] =
\]
\[  O \left ( \frac{n}{p} \times 1 + R_F(n,p) \times Pr[\overline{C} \cup \overline{D}] \right ) =
   O \left ( \frac{n}{p} +\frac{n \log p}{p} \times (Pr[\overline{C}] + Pr[\overline{D}]) \right ) =
\]
\begin{equation}
\label{expectation}
\mbox{(by~(\ref{Cbar}))~~~ }  O \left ( \frac{n}{p} +\frac{n \log p}{p} \times 
    (n^{-1/3} + Pr[\overline{D}]) \right ).
\end{equation}
But
\[ Pr[\overline{D}] =  Pr[\overline{D} \cap C] +  Pr[\overline{D} \cap \overline{C}] \le
  \frac{Pr[\overline{D} \cap C]}{Pr[C]} + Pr[\overline{C}] =
\]
\[ Pr[\overline{D} | C] +  Pr[\overline{C}] = \mbox{~~~(by (\ref{DbarGivenC}) and (\ref{Cbar}))~~~}
O \left (  n^{-(1-2/w)} + n^{-1/3} \right ).
\]
It follows from~(\ref{expectation}) that $E(T(n,p)) = O(n/p)$, and since $\Theta(n/p)$ is
asymptotically optimal, $E(T(n,p)) = \Theta(n/p)$.
\end{proof}

\section{Further remarks}
We have given an asymptotic analysis of the running time of the Saukas-Song selection 
algorithm for coarse grained parallel computers, showing that this algorithm is
efficient but not asymptotically optimal. We have given another algorithm for the selection
problem on coarse grained parallel computers 
with an asymptotically optimal expected running time.

%One can find experimental results for the Saukas-Song selection algorithm in~\cite{SS99}.

Suggestions and corrections offered by the anonymous reviewers are gratefully acknowledged.

%\section{Conflict of interest statement}
%The author states that there is no conflict of 
%interest.

\end{document}